\documentclass[letterpaper,11pt]{article}

\usepackage{prettyref}
\usepackage{amsmath}
\usepackage{amsfonts}
\usepackage{amsthm}
\usepackage{amssymb}
\usepackage{url,ifthen}
\usepackage{srcltx}
\usepackage{multirow}
\usepackage{boxedminipage}
\usepackage[margin=1.125in]{geometry}
\usepackage{nicefrac}
\usepackage{xspace}
\usepackage{graphicx}
\usepackage{xcolor}
\definecolor{DarkGreen}{rgb}{0.1,0.5,0.1}
\definecolor{DarkRed}{rgb}{0.5,0.1,0.1}
\definecolor{DarkBlue}{rgb}{0.1,0.1,0.5}
\usepackage{tikz}
\usetikzlibrary{shapes,arrows}

\usepackage[small]{caption}

\usepackage[pdftex]{hyperref}
\hypersetup{
    unicode=false,          
    pdftoolbar=true,        
    pdfmenubar=true,        
    pdffitwindow=false,      
    pdfnewwindow=true,      
    colorlinks=true,       
    linkcolor=DarkBlue,          
    citecolor=DarkGreen,        
    filecolor=DarkGreen,      
    urlcolor=DarkBlue,          
    pdftitle={},
    pdfauthor={},    
}

\def\draft{0} 

\def\submit{0} 

\ifnum\draft=1 
    \def\ShowAuthNotes{1}
\else
    \def\ShowAuthNotes{0}
\fi

\ifnum\submit=1
\newcommand{\forsubmit}[1]{#1}
\newcommand{\forreals}[1]{}
\else
\newcommand{\forreals}[1]{#1}
\newcommand{\forsubmit}[1]{}
\fi

\ifnum\ShowAuthNotes=1
\newcommand{\authnote}[2]{{ \footnotesize \bf{\color{DarkRed}[#1's Note:
{\color{DarkBlue}#2}]}}}
\else
\newcommand{\authnote}[2]{}
\fi

\newcommand{\Rnote}[1]{{\authnote{Raghu} {#1}}}
\newcommand{\Bnote}[1]{{\authnote{Ben} {#1}}}

\newcommand{\hilight}[1]{#1}

\newtheorem{theorem}{Theorem}[section]

\newtheorem{lemma}[theorem]{Lemma}

\newtheorem{claim}[theorem]{Claim}

\newtheorem{conjecture}[theorem]{Conjecture}
\theoremstyle{definition}
\newtheorem{definition}[theorem]{Definition}

\newtheorem{question}[theorem]{Question}

\newcommand{\chapterref}[1]{\hyperref[ch:#1]{Chapter~\ref{ch:#1}}}

\newcommand{\claimref}[1]{\hyperref[claim:#1]{Claim~\ref{claim:#1}}}

\newcommand{\corollaryref}[1]{\hyperref[cor:#1]{Corollary~\ref{cor:#1}}}
\newcommand{\conjecturelabel}[1]{\label{con:#1}}
\newcommand{\conref}[1]{\hyperref[con:#1]{Conjecture~\ref{con:#1}}}

\newcommand{\definitionref}[1]{\hyperref[def:#1]{Definition~\ref{def:#1}}}

\newcommand{\equationref}[1]{\hyperref[eq:#1]{Equation~\ref{eq:#1}}}
\renewcommand{\eqref}[1]{\hyperref[eq:#1]{(\ref{eq:#1}})}

\newcommand{\factref}[1]{\hyperref[fact:#1]{Fact~\ref{fact:#1}}}

\newcommand{\figureref}[1]{\hyperref[fig:#1]{Figure~\ref{fig:#1}}}

\newcommand{\itemref}[1]{\hyperref[item:#1]{Item~(\ref{item:#1})}}
\newcommand{\lemmalabel}[1]{\label{lem:#1}}
\newcommand{\lemmaref}[1]{\hyperref[lem:#1]{Lemma~\ref{lem:#1}}}

\newcommand{\propref}[1]{\hyperref[prop:#1]{Proposition~\ref{prop:#1}}}

\newcommand{\propositionref}[1]{\hyperref[prop:#1]{Proposition~\ref{prop:#1}}}
\newcommand{\questionlabel}[1]{\label{q:#1}}
\newcommand{\questionref}[1]{\hyperref[q:#1]{Question~\ref{q:#1}}}

\newcommand{\remarkref}[1]{\hyperref[rem:#1]{Remark~\ref{rem:#1}}}

\newcommand{\sectionref}[1]{\hyperref[sec:#1]{Section~\ref{sec:#1}}}
\newcommand{\theoremlabel}[1]{\label{thm:#1}}
\newcommand{\theoremref}[1]{\hyperref[thm:#1]{Theorem~\ref{thm:#1}}}

\usepackage[T1]{fontenc}
\usepackage{kpfonts}
\usepackage{microtype}

\renewcommand{\leq}{\leqslant}
\renewcommand{\le}{\leqslant}
\renewcommand{\geq}{\geqslant}
\renewcommand{\ge}{\geqslant}

\newcommand{\norm}[1]{\lVert#1\rVert}

\newcommand{\R}{\mathbb{R}}

\usepackage{bm}

\newcommand{\iprod}[1]{\langle #1\rangle}

\renewcommand{\epsilon}{\varepsilon}

\newcommand{\remove}[1]{}

\newcommand{\trans}{\top}
\newcommand{\completion}{\textsc{Completion}\xspace}
\newcommand{\psdcompletion}{\textsc{PSD-Completion}\xspace}
\newcommand{\partition}{\textsc{Partition}\xspace}
\newcommand{\coloring}{\textsc{Coloring}\xspace}
\newcommand{\eat}[1]{}
\newcommand{\exactonesat}{\textsc{Exact-one-in-k-SAT}\xspace}
\newcommand{\indepset}{\textsc{Indepedent-Set}\xspace}

\title{Computational Limits for Matrix Completion
}

\author{Moritz Hardt\thanks{IBM Research Almaden. Email: {\tt mhardt@us.ibm.com}.}
\and Raghu Meka\thanks{Microsoft Research. Email: {\tt meka@microsoft.com}.}
\and Prasad Raghavendra\thanks{University of California, Berkeley. Email: {\tt
prasad@cs.berkeley.edu}. Supported by NSF Career Award and the Sloan
Fellowship. } 
\and Benjamin Weitz\thanks{University of California,
Berkeley. Email: {\tt bsweitz@eecs.berkeley.edu}. Supported by the NSF GRFP.}}

\begin{document}
\maketitle

\maketitle

\begin{abstract}
Matrix Completion is the problem of recovering an unknown real-valued low-rank
matrix from a subsample of its entries. Important recent results show that the
problem can be solved efficiently under the assumption that the unknown matrix
is incoherent and the subsample is drawn uniformly at random. Are these
assumptions necessary? 

It is well known that Matrix Completion in its full generality is NP-hard. 
However, little is known if make additional assumptions such as
incoherence and permit the algorithm to output a matrix of slightly
higher rank. In this paper we prove that Matrix Completion remains computationally intractable even
if the unknown matrix has rank~$4$ but we are allowed to output any
constant rank matrix, and even if additionally we assume that the unknown 
matrix is incoherent and are shown $90\%$ of the entries. This result relies on the conjectured hardness
of the $4$-Coloring problem. We also consider the positive semidefinite
Matrix Completion problem. Here we show a similar hardness result under the
standard assumption that $\mathrm{P}\ne \mathrm{NP}.$

Our results greatly narrow the gap between existing feasibility results and
computational lower bounds. In particular, we believe that our results give the
first complexity-theoretic justification for why distributional assumptions are
needed beyond the incoherence assumption in order to obtain positive results.
On the technical side, we contribute several new ideas on how to encode hard
combinatorial problems in low-rank optimization problems. We hope that
these techniques will be helpful in further understanding the computational
limits of Matrix Completion and related problems.
\end{abstract}

\vfill
\thispagestyle{empty}
\pagebreak

\section{Introduction}
Suppose we observe a subset of the entries of an unknown low-rank
matrix~$M$, can we recover the matrix~$M$ knowing this subset alone? This
problem, called Matrix Completion, is of fundamental interest in a number of
fields including statistics, machine learning, signal processing and
theoretical computer science. It is widely applicable to the design of
recommender systems as popularized by the famous Netflix Prize.  We are
interested in understanding the compuational complexity of Matrix Completion.

Much of the theory of Matrix Completion revolves around a beautiful line of
positive results. These results show that under certain assumptions there is a
natural semidefinite relaxation that solves the problem efficiently even if the number
of visible entries is asymptotically much smaller than the total number of
entries~\cite{CandesR09,CandesT10,Recht11}.
Specifically, these feasibility assumptions state that:
\begin{description}
\item[Low rank.] $M$ has rank $k$ where $k$ is typically constant or very slowly growing.
\item[Incoherence.] The row and columns spaces of $M$ 
are \emph{incoherent}. Informally, a subspace $U$ of $\R^n$ is incoherent 
if for every standard basis vector $e_i\in\R^n,$
the Euclidean norm of the projected vector~$P_Ue_i$ is much smaller than~$1.$ 
Here, $P_U$ denotes the orthogonal projection onto the space $U.$ 
\item[Randomness.] Finally, the subset of entries is drawn uniformly at random from~$M$
with a certain sufficient sampling density~$p$.
\end{description}
Among these assumptions the last one is particularly taxing. In most
applications, the algorithm designer cannot choose the subset of revealed
entries. Instead nature determines the subsample, e.g., available user/movie
ratings on Netflix. Often it is argued that without the randomness assumption,
the solution to the problem may no longer be uniquely determined. But rather
than insisting on uniqueness of the solution, it is natural to only require
consistency with the given subset. That is, we only require the solution to
agree with the observed entries. There could be multiple valid solutions.
Moreover, algorithmically two additional relaxations are natural. First, we can attempt
to make the problem easier by allowing some slack in terms of the rank $r>k$
of the solution. Second, we can allow an approximation error on the
observed entries. That is, rather than matching the observable entries exactly
we allow the algorithm to find a solution that is close in Frobenius norm.

Surprisingly, even with these relaxations the status of some 
deceptively simple algorithmic questions remained wide open. For example:
\begin{question}\questionlabel{relaxed}
Given entries of an incoherent rank $4$ matrix, 
can we find a rank $100$ matrix that is approximately consistent with
this set of entries in polynomial time? 
\end{question}
Even though the problem might appear to be very simple, neither upper bounds nor lower
bounds are known. Matrix Completion in its full generality is of course
NP-hard, but no hardness result is known for the problem we just described.
In fact, for small~$k$, the main prior hardness result we are aware of
is due to Peeters~\cite{Peeters96} who showed that given a subset of a
rank~$3$ matrix it is NP-hard to find an exactly consistent matrix of
\emph{equal} rank.\footnote{Several results are known over finite fields, but
Matrix Completion over the reals is of particular interest in applications.} 
However, Peeters' hardness result does not apply to the various relaxations of interest.

The lack of applicable hardness results for Matrix Completion is partially due
to the nature of the problem. Low-rank decompositions over the reals do not seem to 
exhibit the same combinatorial rigidity common to most NP-hard optimization problems. 
This conundrum arises in a number of interesting machine
learning problems such as Sparse PCA and Robust PCA. Indeed, only recently did
Berthet and Rigollet give evidence for computational hardness of the Sparse PCA problem by reducing to the Planted Clique problem in a natural setting~\cite{BerthetR13}. 
For Robust PCA, the hardness result of Hardt and Moitra~\cite{HardtM13} appeals to the conjectured hardness of Small Set
Expansion. 

Our goal is to make progress on understanding the computational complexity of
Matrix Completion in the natural relaxed setting that we described above. 
We show that under a plausible hardness assumption, there is in fact no polynomial time
algorithm that solves the task. An immediate corollary is that even
if we adopt the first two feasibility assumptions, some distributional
assumption on the revealed entries is necessary in order to make Matrix
Completion tractable. 

We also consider a natural variant of Matrix Completion where the unknown
matrix is positive semidefinite and so must be the output matrix. The positive
semidefinite completion problem arises naturally in the context of Support
Vector Machines (SVM). The kernel matrix used in SVM learning must be positive
semidefinite as it is the Gram matrix of feature vectors. But oftentimes the
kernel matrix is derived from partial similarity information resulting in
incomplete kernel matrices. In fact, this is a typical situation in medical
and biological application domains~\cite{TsudaAA03}. In such cases the data
analyst would like to complete the partial kernel matrix to a full kernel
matrix while ensuring positive semidefiniteness. Moreover, since it is often
infeasible to store a dense $n\times n$ matrix, it is desirable to also have a
low-rank representation of the kernel matrix~\cite{FineS02}. 
This is precisely the low-rank
positive semidefinite completion problem. Our results establish strong
hardness results for this problem under natural relexations. In this case we
show that for any $k\ge2$ it is NP-hard to complete a partially given rank~$k$
matrix by a rank~$(2k-1)$ matrix.

\subsection{Our Results}

We will restrict our attention to symmetric $n\times n$ matrices throughout.
As we are proving hardness results, this only makes the results stronger. We
begin with a formal definition of the Matrix Completion problem. Here, we
restrict our attention to the case where both input and output have bounded
coefficients as is the case in most application settings.

\begin{definition}
We define the \hilight{$(k,r,p,\epsilon,c)$-\completion} problem as follows:
\begin{description}
\item[Input:] A matrix $A\in(\R\cup\{\bot\})^{n\times n}$ and a set
$\Omega\subseteq[n]\times[n]$ of size $|\Omega| \geq pn^2$ such that there exists a rank $k$ matrix $M$ with bounded entries \hilight{$|M(i,j)| \leq c$ for all $i$ and $j$}, such that for all $(i,j)\in\Omega$ we have $A(i,j)=M(i,j)$ and for all $(i,j)\not\in\Omega$ we have $A(i,j)=\bot.$
\item[Output:] A rank $r$ matrix $B$ with bounded coefficients \hilight{$|B(i,j)| \leq c$ for all $i$ and $j$}, such that $B$
approximates $A$ with small root-mean-squared error (RMSE): 
$\sum_{(i,j) \in \Omega} |A(i,j) - B(i,j)|^2 \leq \epsilon n .$            
\end{description}
We will use $(k,r,p,c)$-\completion as a shorthand for $(k,r,p,0,c)$-\completion, i.e. exact completion. We also use $(k,r,c)$-\completion as a shorthand for $(k,r,0,0,c)$-\completion.
\end{definition}

To state our first result we introduce the problem of coloring a $k$-colorable
graph with $r$ colors.

\begin{definition}
We define the $(k,r)$-\coloring problem as follows:
\begin{description}
\item[Input:] A $k$-colorable graph $G$.
\item[Output:] An $r$-coloring of the graph $G$.
\end{description}
\end{definition}

Our second theorem will appeal to a closely related variant of the problem in
which the output is an independent set of size $n/r$ rather than an
$r$-coloring.
\begin{definition}
We define the $(k,r)$-\indepset problem as follows:
\begin{description}
\item[Input:] A $k$ colorable graph $G$.
\item[Output:] An independent set of size $n/r$ in the graph $G$.
\end{description}
\end{definition}
Notice that if there exists an $r$-coloring of the graph then one of the color
classes will be an independent set of size $n/r$. Thus, $(k,r)$-\indepset
reduces to $(k,r)$-\coloring. Despite extensive work on algorithms for
$k$-\coloring \cite{Wigderson82,BergerR90,BlumK97,KargerMS98,AroraCC06}, the
problem has remained notoriously hard.  Given a $3$-colorable graph, the best
algorithms \cite{Chlamtac07,KT14} known can only find an independent set of size at
most $n^{1-\Omega(1)}$.  In particular, $(k,r)$-\coloring and
$(k,r)$-\indepset with $k = 4$ and $r = O(1)$ remains hopelessly out of reach
of existing algorithmic techniques. From a complexity standpoint it is
believed that the $(k,r)$-\indepset problem (and hence $(k,r)$-\coloring)
cannot be solved in polynomial time for even $k=4$ and $r = O(1)$. This is
further supported by the work of \cite{DinurS10} who show this to be the case
under a variant of the Unique Games Conjecture (called $2$-to-$1$ Label
Cover) which by now underlies a number of hardness results in complexity
theory.

We will show that assuming $(k,r)$-\coloring and $(k,r)$-\indepset are hard
for $k=4,r=O(1)$, the Matrix Completion problem is hard in a range of natural
parameters even on incoherent matrices and even if most entries are revealed.
To make the theorems precise we state the assumption concretely and give a
formal definition of incoherence.

\begin{conjecture}\conjecturelabel{coloring}
The $(k,r)$-\coloring problem is not in P for any $r \geq k\ge 3$ and $r = O(1)$.  
\end{conjecture}

\begin{conjecture}\conjecturelabel{indepset}
The $(k,r)$-\indepset problem is not in P for any $r \geq k \ge 3$ and $r = O(1)$.  
\end{conjecture}

The coherence of a matrix is defined as follows.

\begin{definition}
A symmetric $n\times n$ matrix $M$ of rank~$k$ has coherence $\mu$ if there
exists a singular value decomposition $M=U\Sigma V^\trans$ such that for every
standard basis vectors $e_i\in\R^n$ we have that $\|e_i^\trans U\|_2 \le
\sqrt{k\mu/n}$ and $\|e_i^\trans V\|_2\le \sqrt{k\mu/n}.$
\end{definition}

Note that \conref{coloring} is weaker than \conref{indepset}. With the above definitions we have the following results.

\Bnote{Changed the statements of Theorems 1.8 and 1.9 to accurate settings for $\epsilon$ and $k$}

\Bnote{I think we had $k \geq 4$ as an artifact of referencing the Dinur coloring hardness proof}

\begin{theorem}\theoremlabel{maincomplete}
Assume \conref{coloring}. Then, for any constants $c \geq 1$, \hilight{$k\ge 3$} and $r>k$, there is no polynomial time algorithm that solves the \hilight{$(k,r,0.9,c)$-\completion} problem on matrices of
coherence $\mu\le O(1)$. Further, for all \hilight{$1/2 > \epsilon > 0$}, the same conclusion
holds even if we are only required to compute a rank $r$ matrix which
approximates each entry with additive error at most $\epsilon$. 
\end{theorem}

In most practical scenarios it suffices to look for a low-rank completion with small RMSE error. 
Our next result addresses this situation.

\begin{theorem}\theoremlabel{maincompleteapprox}
Assume \conref{indepset}. Then, for any constants $k\ge 4$, $r>k$, and \hilight{$0 < \epsilon < 1/(2cr)^2$}, there is no polynomial
time algorithm that solves the \hilight{$(k,r,0.9,\epsilon,c)$-\completion} problem on matrices of
coherence $\mu\le O(1).$ 
\end{theorem}

This result should be contrasted with positive results showing that
$(k,k,O(k\mu(\log^2 n))/n)$-\completion is easy so long as the
entries are revealed randomly~\cite{Recht11}.


\paragraph{Positive Semidefinite Completions.}   
We define the \hilight{$(k,r,p)$-\psdcompletion} problem the same way we
defined \hilight{$(k,r,p)$-\completion} except that we drop the bound on the
coefficients and additionally require that both $M$ and $B$ must be positive
semidefinite. Our result here is incomparable to the previous one and it
relies on the standard NP-hardness assumption.

\begin{theorem}\theoremlabel{psdcompletemain}
Assume that $\mathrm{P}\ne\mathrm{NP}.$ Then for every even $k\ge 2$ there is no
polynomial time algorithm that solves the $(k,2k-1,0.9)$-\psdcompletion problem.
\end{theorem}

This theorem strengthens a recent result by E.-Nagy et al.~\cite{ENagyLV13}
who showed that $(k,k,0)$-\psdcompletion is NP-hard for every $k\ge2.$ 

We also prove a version of \theoremref{psdcompletemain} for approximate completion:
\begin{theorem}\theoremlabel{psdapproxcompletemain}
Assume that $\mathrm{P}\ne\mathrm{NP}$. Then for every even $k\ge 6$ and $\epsilon < O(k^{-5})$, there is no polynomial time algorithm that solves the $(k,2k-1,0.9)$-\psdcompletion problem 

\hilight{even if the output matrix only approximates each entry with additive error at most $\epsilon$.}
\end{theorem} 

\subsection*{Acknowledgments}

We are very grateful to Phil Long for insightful early contributions to this
work. In fact, he first conjectured that  $(k,r)$-\completion should be as hard
as $(k,r)$-\coloring as established by \theoremref{maincomplete}. The authors
also thank the Simons Institute for the Theory of Computing at Berkeley for
its hospitality.

\subsection{Further Related Work}

There have been several hardness results for Matrix Completion over finite
fields drawing on its connection to problems in coding theory.  See, for
example, the discussion in~\cite{HarveyKY06,TanBD12}. The Matrix Completion
problem over the reals seems to behave rather differently and techniques do
not seem to transfer from the finite field case. 

PSD completions are also natural objects in discrete optimization and the study of
the geometry of graphs. We refer the reader to the recent work of E.-Nagy,
Laurent and Varvitsiotis~\cite{ENagyLV13} for a more extensive discussion of 
related work in this area.

\subsection{Proof Overview}
We now give a highlevel outline of our proofs.
\subsubsection{Matrix Completion}
While the hardness assumption in our reduction (\conref{coloring}) is similar in spirit to that of Peeters' original reduction, our proof works in a very different manner.


Let $G = (V,E)$ be a graph with $|V| = n$ and $|E| = m$. Now define the $n \times n$ partial matrix $P_G$ such that $P_G(i,i) = 1$ for every $i \in [n]$, and $P_G(i,j) = 0$ if $(i,j) \in E$. The intuition behind this reduction is that, if $G$ is $k$-colorable with coloring function $f: V \rightarrow [k]$, then 
\[M_f = \sum_{i \in [k]} 1_{f^{-1}(i)}1_{f^{-1}(i)}^T\]
is a rank-$k$ completion of $P_G$. Peeters \cite{Peeters96} showed how to gadgetize a
graph $G$ so that these were the \emph{only} rank-$k$ completions of
$P_G$. However, the gadgets in that work are unable to force any
structure on completions of rank higher than $k$. To decode colorings from high rank completions, we will consider a special factorization of the completion. The row vectors of this factorization will have bounded norm, which will allow us to cover them with a constant number of small balls in $\mathbb{R}^r$. If the balls are small enough, then any two vectors that lie in the same ball cannot have zero dot product, so their corresponding vertices cannot have an edge in $G$. We then use the balls to color the vertices. 

The above argument works when we look at exact completions (or those with entry-wise error bounds). To obtain our main result, \theoremref{maincompleteapprox}, we focus on more general structure of any low-rank completion, in this case the existence of large non-zero rectangles. We will prove, under some mild assumptions on any approximate (in RMSE) completion $M$ of $P_G$, that $M$ has a large non-zero square, which corresponds to an independent set in the graph $G$.
\Bnote{Changed to reflect coloring reduction rather than IS reduction.}
  

\subsubsection{Positive Semidefinite Matrix Completion}
We give two reductions for the $(k,r)$-\psdcompletion problem: one from the
\partition problem and one from a constraint satisfaction problem
\exactonesat. The first reduction has the advantage of being simple but only
works for \emph{exact} completion. Our second reduction is more involved but
is much more robust and it works even when we allow for errors and gives us the theorem on approximate completions from the introduction. 

Consider an instance of the $(k,r)$-\psdcompletion problem with input $A\in(\R\cup\{\bot\})^{n\times n}$. Our goal is to find a PSD matrix $B$ which agrees with $A$ on the set of non-$\bot$ entries. Now, recall that a characterization of PSD matrices is that a $n \times n$ matrix $B$ is PSD if and only if it can be factored as $B = UU^\top$ for some matrix $U$. If we let $u_1,\ldots,u_n$ be the rows of the matrix $U$, then we have $B_{ij} =\iprod{u_i,u_j}$. The vectors $u_1,\ldots,u_n$ are called the \emph{Gram vectors} of $B$. In the context of \psdcompletion, the revealed entries of $A$ place equality constraints on the inner products of the Gram vectors: 
$$\iprod{u_i,u_j} = A_{ij}, \;\; \text{if $A_{ij} \neq \bot$}.$$
Moreover, these constraints completely characterize the problem and finding a rank $r$ solution for the completion problem is equivalent to finding a set of vectors $u_1,\ldots,u_n \in \R^r$ satisfying the above constraints. We will adopt this perspective in our reductions and view the partial matrix as a list of such inner-product constraints. 

We design constraints to simulate $\pm 1$ variables which we can then
use as \emph{gadgets} to reduce from many different problems. For the \partition problem, we follow an idea proposed in \cite{ENagyLV13} and associate every item in the partition with a two-dimensional basis, and constrain that the $(i+1)$th basis is a $\theta$-rotation of the $i$th basis (including the first and $n$th bases), where $\theta$ depends on the element $a_i$ in the \partition problem. This creates a cyclic dependence on the rotations of the bases that forces the total sum of the rotations to be an integer multiple of $2\pi$. However, the important things to note are that these rotations can be in one of two direction: clockwise or counter-clockwise, and if the angles are small enough then the sum of the rotations must be zero. Thus we find a partition based on which rotations went clockwise and which went counter-clockwise. By constraining sums of basis vectors in addition to the basis vectors themselves, we can force the same rotational structure in three dimensions as in two, yielding the gap. 

For the \exactonesat problem, we similarly associate every variable with a basis and use the inner product constraints to force these bases to be special rotations of a reference gadget. We interpret the variable as being $+1$ or $-1$ depending on if the rotation is a "clockwise" or "counter-clockwise" rotation. Because the relations of \exactonesat are linear, i.e. the sum of the values of the variables in each clause is exactly $(k-2)$, it is easy to force satisfying assignments. See Section \ref{cspreduct} for a more thorough description and the appendix for the full details.


\section{Hardness for Matrix Completion}
In this section we show that the matrix completion problem is hard even with relaxed rank constraints and allowing for approximate completions. In particular, we will give a reduction to prove \theoremref{maincompleteapprox}. We defer the proof of \theoremref{maincomplete} to the appendix.

Let $G = (V,E)$ be a graph with $|V| = n$ and $|E| = m$. Now define the partial matrix $P_G \in (\R \cup \bot)^{n \times n}$ as follows:
\begin{equation*}
  P_G(i,j) = \begin{cases} 1 &\text{if $i = j$}\\
0 &\text{if $(i,j) \in E$}\\
\bot &\text{otherwise}    
  \end{cases}.
\end{equation*}
As described in the introduction, if $G$ is $k$-colorable with coloring function $f: V \rightarrow [k]$, then 
\[M_f = \sum_{i \in [k]} 1_{f^{-1}(i)}1_{f^{-1}(i)}^T\]
is a rank-$k$ completion of $P_G$. Note that $M_f$ has coherence $\mu = \frac{n}{k}(\min_i |f^{-1}(i)|)^{-1}$. However, we may assume that there is a perfectly balanced coloring of $G$, for example by copying the graph $k$ times. Thus we can take $M_f$ to have coherence exactly $\mu = 1$. We next prove that under some mild additional assumptions any approximate low-rank completion $M$ of $P_G$ yields a large independent set of $G$. 

\Bnote{Here I explain the use of Linial's lemma}

Our technique relies on a lemma in \cite{Linial07} that guarantees the existence of a good factorization of low rank matrices, in the sense that the norms of the row and column vectors of the factorization are small. We state the lemma here for completeness:
\begin{lemma}\lemmalabel{Linial}
Let $M$ be a matrix with rank $r$. Then there exists an $r$-dimensional factorization $M = XY^T$ such that every row vector of $X$ and $Y$ has norm at most $(cr)^{1/4}$, where $c = \max_{ij} |M(i,j)|$.
\end{lemma}
The proof of this statement in its original paper was given in a nonconstructive manner using John's theorem from convex analysis, but such a good factorization can be found in polynomial time using semidefinite programming. See the appendix for the details.
\begin{lemma}\lemmalabel{islemma}
Let $G = (V,E)$ be a graph and define $P_G$ as above with $\Omega \subseteq [n] \times[n]$ the set of revealed entries. Let $M$ be a rank $r$ matrix such that 
\[\sum_{(i,j) \in \Omega} (M(i,j)-P_G(i,j))^2 \leq \epsilon n\]
and $|M(i,j)| \leq c$. Then $G$ has an independent set $T$ of size at least
\[|T| \geq \frac{(1-4(cr)^2\epsilon)n}{r\sqrt{\pi}(8\sqrt{cr})^r}.\]
Moreover, there is a randomized polynomial time algorithm for finding such an independent set given $M$.
\end{lemma}
\begin{proof}
By a simple averaging argument, there can be only $\frac{\epsilon}{\delta^2}n$ entries of $M$ that are different from $P_G$ by more than $\delta$. Thus there are at least $(1-\epsilon/\delta^2)n$ rows and columns such that $|M(i,j) - P_G(i,j)| \leq \delta$ for any row $i$ and column $j$. Let $M'$ be this submatrix of $M$. Certainly rank$(M') \leq \text{rank}(M) = r$, so \lemmaref{Linial} tells us we can find a factorization $XY^T = M$ with row vectors $u_i$ and $v_i$ such that $\|u_i\|,\|v_j\| \leq (cr)^{1/4}$ for all $i,j \in [n]$. Let $\theta(u_i,v_j)$ denote the angle between vectors $u_i$ and $v_j$. Since $u_i \cdot v_i \geq 1-\delta$ for all $i \in [n]$, we derive
\[\cos(\theta(u_i,v_i)) \geq \frac{(1-\delta)}{\sqrt{cr}} \text{ and } \|u'_i\|,\|v'_i\| \geq (1-\delta)(cr)^{-1/4}.\]
Now in order to find an independent set, we have to look for entries with $M(i,j) > \delta$ to be assured that indeed $(i,j) \notin E$. From the bound on the norms of $u_i$ and $v_j$, if $\cos\theta(u_i,v_j) > \delta\frac{\sqrt{cr}}{(1-\delta)^2}$, then $M(i,j) > \delta$. In order to capture these points, we will pick the points in a random cone. Let $\phi$ denote the angle such that $\cos \phi = \delta\frac{\sqrt{cr}}{(1-\delta)^2}$. Our random procedure to find $T$ is
\begin{itemize}
\item Normalize $\tilde{u}_i = u_i/\|u_i\|$ and $\tilde{v}_i = v_i/\|v_i\|$.
\item Pick a random unit vector $x \in \R^r$. 
\item For every $i \in S$, if $\tilde{u}_i \cdot x > \cos (\phi/2)$ and $\tilde{v}_i \cdot x > \cos (\phi/2)$ then put $i \in T$.
\end{itemize}
For $i,j\in T$, since $\theta(\tilde{u}_i,x) < \phi/2$ and $\theta(\tilde{v}_i,x) < \phi/2$, by triangle inequality, $\theta(\tilde{u}_i,\tilde{v}_i) < \phi$. As noted above, since $\cos \theta(\tilde{u}_i,\tilde{v}_i) > \cos \phi$, we get $u_i \cdot v_j > \delta$, thus $M(i,j) > \delta$ and so $P_G(i,j) > 0$. We bound $|T|$ by checking the probability that $i$ is placed in $T$. For each $i$, let $w_i$ be the angle bisector of $\tilde{u}_i$ and $\tilde{v}_i$ and define
\[A_i = \left\{x: \|x\|=1, \theta(x,w_i) < \frac{1}{2}\left(\phi - \theta(\tilde{u}_i,\tilde{v}_i)\right)\right\}.\]
We will show that if $x \in A_i$ was chosen as our random vector, then $i \in T$. This implies that the probability that $i \in T$ is at least area$(A_i)/\text{area}(S^{r-1})$. For $A_i$ to have positive area, we need $\delta$ small enough that $\phi > \theta(\tilde{u}_i,\tilde{v}_i)$. To this end, pick $\delta = 1/2cr$. It is a standard argument that the area of $A_i$ is bounded below by the $(r-1)$-volume of a sphere with radius
\[b_i = \sin\left(\frac{\phi - \theta(\tilde{u}_i,\tilde{v}_i)}{2}\right)\]
Now noting that $\cos \phi = \delta\sqrt{cr}/(1-\delta)^2$ and $\cos \theta_i \geq (1-\delta)/\sqrt{cr}$ and using a Taylor Series approximation
\begin{align*}
\sin\frac{1}{2}\left(\cos^{-1}\left(\frac{x^2/2}{x(1-x^2/2)^2}\right) - \cos^{-1}\left(x(1-x^2/2)\right)\right) &\geq \frac{x}{4} - O(x^3) \geq \frac{x}{8}
\end{align*}
where $x = 1/\sqrt{cr}$ and the last inequality follows as long as $\sqrt{cr} \geq 1$. Now 
\[\frac{\text{area}(A_i)}{\text{area}(S^{r-1})} \geq \frac{b_i^{r-1}}{r\sqrt{\pi}},\]
and thus $i \in T$ with probability at least $1/r\sqrt{\pi}(8\sqrt{cr})^r$. Now using linearity of expectation,
\[E[|T|] \geq \frac{(1-\epsilon/\delta^2)n}{2r\sqrt{\pi}(8\sqrt{cr})^r} \geq \frac{(1-4(cr)^2\epsilon)n}{r\sqrt{\pi}(8\sqrt{cr})^r}.\]
\end{proof} 
The above reduction produces a partial matrix $P_G$ that has $|V| + |E|$ revealed entries, which could be much less than $0.9n^2$ if the graph $G$ is sparse. However, we can simply pad the matrix $P_G$ with zeros, i.e. output the $10|V| \times 10|V|$ matrix
\[P'_G = \left[\begin{tabular}{cc} $P_G$ & $0$ \\ $0$ & $0$\end{tabular}\right].\]
Combined with \lemmaref{islemma}, the above implies \theoremref{maincompleteapprox}. \Bnote{Put the section on padding with zeros back in}
\section{Hardness for Positive Semidefinite Matrix Completion}

To prove hardness for the $(k,r,p,\epsilon)$-\completion problem we appealed to a conjectured coloring hardness. In this section we show that this assumption can be weakened to the usual $\mathrm{NP}$-hardness if the matrices under consideration are positive semi-definite. In particular, in this section we prove \theoremref{psdcompletemain}. We first present the hardness for the exact completion problem with $\epsilon = 0$ with a reduction from \partition. We then sketch the second reduction from \exactonesat that is capable of handling errors on the constraints and proves \theoremref{psdapproxcompletemain}. The full details on the second reduction are deferred to the appendix.
\subsection{Exact Completion}
Our reduction is similar to Theorem 3.3 in \cite{ENagyLV13}, but with extra constraints to retain structure in higher rank completions. We will reduce from the partition problem, i.e. given numbers
$a_1,\dots,a_n$, find a set $I \subseteq [n]$ such that $\sum_{i \in I}
a_i = \sum_{i \notin I} a_i$. We will reduce this problem to
$(2,3)$-\psdcompletion and amplify the gap. Recall that a partial PSD matrix is equivalent to a list of inner product constraints. Given an instance $(a_1,\dots,a_n)$, we will output a set of constraints on $3n$ different vectors. These vectors will be indexed by $I = [n] \times [3]$. Assume without loss of generality that $\sum_i a_i = 1$. Now constrain
\begin{itemize}
\item $u_s \cdot u_s = 1$ for all $s \in I$
\item $u_{(i,1)} \cdot u_{(i,2)} = 0$ for all $i \in [n]$.
\item $u_{(i,3)} \cdot u_{(i,1)} = u_{(i,3)}\cdot u_{(i,2)} = \frac{1}{\sqrt{2}}$. Equivalently, $u_{(i,3)} = \frac{1}{\sqrt{2}}(u_{(i,1)} + u_{(i,2)})$. 
\item $u_{(i,1)}\cdot u_{(i+1,1)} = u_{(i,2)} \cdot u_{(i+1,2)} = u_{(i,3)} \cdot u_{(i+1,3)} = \cos a_i$ for all $i \in [n]$, where addition is performed modulo $n$.
\end{itemize}
The intuition here is that in a rank-$2$ decomposition of this matrix,
for every $i$ $\{u_{(i,1)},u_{(i,2)}\}$ is an orthonormal basis, and
the $(i+1)$st basis is an angle $a_i$-rotation of the $i$th basis.
This rotation can be in one of two directions, clockwise or
counter-clockwise.  However since the $1$st basis is an $a_n$-rotation
of the $n$th basis, after rotating by every $a_i$ we must be back
where we started. Since $\sum_i a_i = 1 < 2\pi$, this means that the
total rotation must be zero, so we partition the $a_i$ based on
whether the corresponding rotation was clock-wise or counterclockwise. The additional constraints on the sums of basis vectors will force this structure even in a rank-$3$ decomposition. 
\begin{lemma} \label{lem:twodimensions} There is a set of vectors $\{u_s\}_{s\in I}$ lying in $\R^3$ satisfying the above constraints if and only if there is a set of such vectors lying in $\R^2$. 
\end{lemma}
\begin{proof}
Let $\{u_s\}_{s \in I}$ be a set of vectors lying in $\R^3$ satisfying the constraints. For each $i$, there is an orthogonal transformation $Q_i$ that maps $Q_i(u_{(i,1)}) = u_{(i+1,1)}$, $Q_i(u_{(i,2)}) = u_{(i+1,2)}$. Writing $Q_i$ in the basis $\{u_{(i,1)},u_{(i,2)}, u_{(i,1)} \times u_{(i,2)}\}$, and accounting for the constraints of the $\{u_s\}_{s\in I}$, we have
\[Q_i = \left[\begin{tabular}{ccc} $\cos a_i$ & $x$ &  \\ $-x$ & $\cos a_i$ & $A$ \\ $y$ & $z$ & \end{tabular}\right]\]
but $Q_i$ has orthogonal columns, which implies that $yz = 0$, so either $y = 0$ or $z = 0$. But $Q_i$ also has columns of norm $1$, so $x = \pm\sin a_i$, which implies that both $y$ and $z$ are zero. Thus 
\[Q_i = \left[\begin{tabular}{ccc} $\cos a_i$ & $\pm\sin a_i$ & $0$ \\ $\mp \sin a_i$ & $\cos a_i$ & $0$ \\ $0$ & $0$ & $1$\end{tabular}\right].\]
This implies that $\{u_{(i,1)}, u_{(i,2)}\}$ and $\{u_{(i+1,1)}, u_{(i+1,2)}$ lie in the same plane. Repeating the argument we get that $\{u_{(i,1)},u_{(i,2)}\}_i$ lie in the same plane for every $i$.
\end{proof}
\begin{lemma} \label{lem:partitionpsd} There is a partition of $(a_1,\dots,a_n)$ if and only if there is a set of vectors satisfying the constraints lying in $\R^2$.
\end{lemma}
\begin{proof}
First, assume there is a partition $(I,\overline{I})$ of $[n]$ and set 
\[\theta_k = \sum_{i \in I,i<k} a_i - \sum_{i \notin I,i<k} a_i\]
and by convention $\theta_1 = 0$. Note that by the definition of $\theta$ and partitions, $\theta_n = \pm a_n$. Now set $u_{(i,1)} = e_1\cos \theta_i + e_2\sin \theta_i$, $u_{(i,2)} = e_1\cos \theta_i - e_2\sin \theta_i$, and $u_{(i,3)} = \frac{1}{\sqrt{2}}(u_{(i,1)} + u_{(i,2)})$ for every $i \in [n]$. Then $u_{(i,1)} \cdot u_{(i+1,1)} = u_{(i,2)} \cdot u_{(i+1,2)} = \cos (\theta_i - \theta_{i+1}) = \cos a_i$ for every $i < n$. Finally, since $\theta_n = \pm a_n$ and $\theta_1 = 0$, $u_{(n,1)}$ and $u_{(n,2)}$ are at an angle $a_n$ with $u_{(1,1)}$ and $u_{(1,2)}$ respectively, so $u_{(n,1)} \cdot u_{(1,1)} = u_{(n,2)} \cdot u_{(n,1)} = \cos a_n$.

Conversely, suppose there is a set of vectors lying in $\R^2$ satisfying the constraints. Since the vectors are unit vectors and $u_{(i,1)} \cdot u_{(i+1,1)} = \cos a_i$, we know $u_{(i+1,1)}$ makes an angle $a_i$ with $u_{(i,1)}$, so for every $i \in [n]$,
\[u_{(i,1)} = \cos \left(\sum_{j=1}^{i-1} s_ja_j\right)u_{(1,1)} + \sin\left(\sum_{j=1}^{i-1} s_ja_j\right)u_{(1,2)}\]
where $s \in \{+1,-1\}^{n}$. Finally, since $u_{(n,1)}$ is at an angle $a_n$ with $u_{(1,1)}$, we have
\[u_{(1,1)} = \cos \left(\sum_{j=1}^{n} s_ja_j\right)u_{(1,1)} + \sin\left(\sum_{j=1}^{n} s_ja_j\right)u_{(1,2)}\]
Since $\sum_j a_j = 1 < 2\pi$, we must have $\sum_{j=1}^n s_ja_j = 0$.
Hence the set $I = \{i: s_i = +1\}$ yields a solution to the
\partition problem on the instance $(a_1,\ldots, a_n)$. 
\end{proof}
The $\mathrm{NP}$-hardness of $(2,3)$-\psdcompletion follows from
\prettyref{lem:twodimensions} and \prettyref{lem:partitionpsd}. There's a simple amplification one can do to prove hardness for
$(k,2k-1)$-\psdcompletion for any even $k$. Simply take the matrix $A$ from the $(2,3)$-\psdcompletion reduction and output the matrix
\[M = \left[\begin{tabular}{cccc} $A$ & $0$ & $\cdots$ & $0$ \\ $0$ & $A$ & $\cdots$ & $0$ \\ $\vdots$ & $\vdots$ & $\ddots$ & $\vdots$ \\ $0$ & $0$ & $\cdots$ & $A$\end{tabular}\right]\]
where $A$ appears $k/2$ times. Then any completion of $M$ of rank at most $2k-1$ will restrict to a completion of $A$ of rank at most $3$. We can also pad with additional zeros to boost the number of revealed entries up to $0.9n^2$. This completes the proof of \theoremref{psdcompletemain}.
\Bnote{Added a brief descriptive section on tolerating errors in the PSD case}
\subsection{Tolerating Errors}
\label{cspreduct}
In this section we sketch the reduction in the proof of \theoremref{psdapproxcompletemain}. Because \partition is only $\mathrm{NP}$-hard in a weak sense (there exists an algorithm which runs in time polynomial in the size of the weights $a_1,\ldots,a_n$), we require a different starting problem to prove hardness while tolerating errors in the constraints. 
\begin{theorem}\theoremlabel{psdsoundness}
For every constant $k\geq 3$, constant $r$ with $2k \leq r \leq 4k-1$, and $\epsilon < O(r^{-5})$, there is a reduction from \exactonesat that, given an instance $\Phi$, outputs a partial matrix $P_\Phi$ with the following property: If there is a rank $r$ matrix $M$ such that $|M(i,j) - P_\Phi(i,j)| < \epsilon$ whenever $P_\Phi(i,j) \neq \bot$, then $\Phi$ is satisfiable. Furthermore, if $\Phi$ is satisfiable, there is a rank $k$ completion of $P_\Phi$.
\end{theorem}
There are two main components to the reduction in \theoremref{psdsoundness}. The first is the \emph{variable gadget}, a set of constraints that forces only two configurations for a set of vectors, and we can interpret these configurations as being a $+1$ or $-1$ assignment to the variable. The second is the \emph{clause gadget}, a set of constraints designed to force the interpreted assignment to be satisfying. 

For each variable in the instance $\Phi$, the variable gadget is a set of constraints that creates a $2k$-dimensional orthonormal basis. There is also a special "reference basis" that we use as a reference point because inner product constraints are invariant to rotations. For each variable, we constrain its basis to be a special rotation of the reference basis. The rotation is special in the sense that it is a set of identical rotations in $k$ pairs of two-dimensional subspaces. Because a rotation in two dimensions has exactly two configurations, rotate clockwise or rotate counter-clockwise, there are only two possible rotations of the variable's basis. We interpret each of these rotations as setting the variable to $+1$ or $-1$. 

For each clause in $\Phi$, the clause gadget is a set of constraints that are intended to construct a vector whose $i$th coordinate is the value (either $+1$ or $-1$) of the $i$th variable appearing in $\Phi$. Finally, we constrain that the sum of the elements of $\Phi$ is exactly $(k-2)$. This forces exactly one of the coordinates of the vector to be $-1$, so the variables must be set to a satisfying assignment.

It is not obvious how we are able to force such specific structure on the rotations in the variable gadget even when the ambient dimension gets as high as $4k-1$. By constraining dot products of sums of basis vectors in the variable gadget we are able to resolve this problem. The full details and proof of \theoremref{psdsoundness} are located in the appendix for the interested reader.


\section{Conclusion and Open Problems}

Our goal was to narrow the gap between existing positive results on Matrix
Completion and computational lower bounds. For a hardness result to be
compelling it must account for natural algorithmic relaxations. We showed that
several relaxations that are natural from an algorithmic machine learning
point of view do not make the problem easier. From a complexity theoretic
perspective, these are the first hardness of approximation results for 
Matrix Completion over the reals. A consequence of our work is that the
popular incoherence assumption by itself is not sufficient to make the problem
tractable. An interesting question is if conversely the assumption of
uniformly random entries by itself already makes the problem easy.

\begin{question}
Is Matrix Completion hard when the observed entries are chosen randomly, but
the observed matrix is \emph{not} incoherent?
\end{question}

Another challenging question is to determine the precise hardness threshold for
$(k,r)$-completion.

\begin{question}
For any $k\ge 3,$ what is the largest $r\ge k$ such that $(k,r)$-\completion is hard? Can we
find a matching algorithm?
\end{question}

Resolving this question will likely require progress on both lower bounds and
algorithms. Deceptively simple algorithmic questions are still open such as
the following.

\begin{question}
We know that $(3,\sqrt{n})$-\coloring is easy~\cite{Wigderson82}. 
Is $(3,\sqrt{n})$-\completion easy?
\end{question}

\bibliographystyle{alpha}
\bibliography{moritz,matrixcompletion_rough}

\appendix
\section{\hilight{Finding Good Factorizations of Low Rank Matrices}}
In this section we will prove a constructive version of
\lemmaref{Linial}, restated below for convenience.

\begin{lemma}(Lemma 2.1 restated)
Let $M$ be a matrix with rank $r$. Then there exists an $r$-dimensional factorization $M = XY^T$ such that every row vector of $X$ and $Y$ has norm at most $(cr)^{1/4}$, where $c = \max_{ij} |M(i,j)|$.
\end{lemma}

In \cite{Linial07} this lemma was proven in a
non-constructive manner using John's Theorem from convex analysis.
However, as the authors therein note, there is a simple semi-definite
program (SDP) one can write to compute a good factorization.
Unfortunately, this SDP may give a factorization that is not
$r$-dimensional, but instead up to $(m+n+1)$-dimensional. Here we give
a simple self-contained algorithm to compute the decomposition
$M = XY^T$. 

\begin{proof}[Proof (of an algorithm to construct decomposition as in \lemmaref{Linial})]
Let $M$ be a rank $r$ matrix, and let $c = \max_{ij} |M(i,j)|$. To construct a good $r$-dimensional factorization of $M$, we proceed in a few steps:
\begin{enumerate}
\item Compute a factorization $M = UV^T$ such that $U$ and $V$ have
	short row vectors using semi-definite programming. This
	factorization is not necessarily $r$-dimensional. 
\item Project every row of $V$ onto the row space of $U$ to get $V'$. 
\item Factor $V' = YB$, where $B$ is an orthonormal basis for the row space of $V'$.
\item Output $X = UB^T$ and $Y$. This will be an $r$-dimensional factorization with short rows. 
\end{enumerate}
For the first step, consider the following SDP with variables
$\{u_i\}_{i=1}^m$, $\{v_j\}_{j=1}^n$ and $\eta$.
\begin{align*}
	\text{ Minimize  }  & \eta \\
	\text{ Subject to } \\
	& u_i \cdot v_j = M(i,j) & \forall i \in \{1,\ldots, m\} j \in
	\{1,\ldots, n\} \\
	& u_i \cdot u_i \leq \eta & \forall i \in \{1,\ldots, m\} \\
	& v_j \cdot v_j \leq \eta & \forall j \in \{1,\ldots, n\} \\
\end{align*}

This SDP computes the row vectors of two matrices $U$ and $V$ such
that $M = UV^T$ and minimizes their row norms. The lemma in
\cite{Linial07} guarantees that there exists a factorization with the
norm of every row vector of $U$ and $V$ bounded by $(cr)^{1/4}$, and
thus the factorization found by the SDP is at least this good.
However, since there are $m+n+1$ variables, the vectors returned could
be up to $(m+n+1)$-dimensional. 

Let $V'$ be the matrix whose rows are the rows of $V$ projected onto
the row space of $U$. Note that we still have $M = U(V')^T$, and
projections can only decrease the lengths of the rows of $V$.  We
make the following claim.

\begin{claim}
Rank of $V'$ is equal to rank of $M$, i.e., $r$.
\end{claim}
\begin{proof}
Let
$r_u$ and $r_v$ denote the ranks of $U$ and $V'$, respectively. Clearly $r_v \geq r$
because $V'$ is a factor of $M$. 
%

Now pick linearly independent
subsets $E$ of $r_u$ rows of $U$ and $F$ of $r_v$ rows of $V'$. Then
$EF^T$ is a submatrix of $M$, and thus must have rank at most $r$.
However, we will argue below that the rank of $EF^T$ must be equal to $r_v$,
which shows that $r_v \leq r$.  We already know that $r_v \geq r$,
which implies that $r_v = r$ as desired.

To see that rank of $EF^T$ is $r_v$, recall that the row space of $V'$
is contained in the row space of $U$.  Since $E$ and $F$ form a basis for
row-spaces of $U$ and $V'$,  row space of $F$ is contained in the row
space of $E$.  So for any vector $x$, $x^TF$ lies in the row space of
$F$, and thus in the row space of $E$. This implies that whenever $x^TF$
is nonzero, $E(F^Tx)$ must be nonzero as well.  So the rank of $EF^T$ is the same as the rank of $F^T$, i.e. $r_v$. 
\end{proof}

Now since $V'$ has rank $r$, there is a factorization $V' = YB$, where
the rows of $B$ are an $r$-dimensional orthonormal basis for the row
space $V'$, and $Y$ is the matrix whose rows are the coordinates for
the rows of $V'$ in the basis $B$. Note that since $B$ is orthonormal,
the norms of rows of $Y$ are the same as the norms of rows of $V'$.
Let $X = UB^T$. Note that the rows of $X$ are given by the the
projections of the rows of $U$ onto the basis $B$. Since $B$ is
orthonormal, this projection cannot increase the norms of the rows.

Observe that $XY^T = UB^T(V')^T = M$.  Moreover $XY^T$ is an
$r$-dimensional factorization for $M$ because $B$ is an
$r$-dimensional basis.  Finally, the norms of rows of $X$ and $Y$ are
smaller than the norms of rows of $U$ and $V$ given by the
semidefinite program, as we only decreased the lengths at each step
described above.
\end{proof}

\section{Hardness of Matrix Completion}
In this section we will prove \theoremref{maincomplete}. The basic setup is the same as in the proof of \theoremref{maincompleteapprox}. Let $G = (V,E)$ be a graph with $|V| = n$ and $|E| = m$. Now define the partial matrix $P_G \in (\R \cup \bot)^{n \times n}$ as follows:
\begin{equation*}
  P_G(i,j) = \begin{cases} 1 &\text{if $i = j$}\\
0 &\text{if $(i,j) \in E$}\\
\bot &\text{otherwise}    
  \end{cases}.
\end{equation*}
Then, as before, if $G$ is $k$-colorable with coloring function $f: V \rightarrow [k]$, then 
\[M_f = \sum_{i \in [k]} 1_{f^{-1}(i)}1_{f^{-1}(i)}^T\]
is a rank-$k$ completion of $P_G$. Similar to the proof \theoremref{maincompleteapprox} we can also assume that $M_f$ has coherence $\mu=1$. We next prove that under some mild additional assumptions any entry-wise-approximate low-rank completion $M$ of $P_G$ yields a coloring of $G$ with few colors (as opposed to just finding a large independent set as done earlier). 
\begin{lemma}\lemmalabel{coloringlemma}
Let $G$ be a graph and define the partial matrix $P_G$ as above. Let $M$ be a rank-$r$ matrix with bounded coefficients $|M(i,j)| \leq c$ for all $i,j \in [n]$ such that $M$ approximates $P_G$, i.e. $|M(i,j) - P_G(i,j)| < \epsilon$ whenever $P_G(i,j) \neq \bot$. If $\epsilon < 1$ then there is an algorithm that colors $G$ using only
\[\chi(G) \leq \left(\frac{4\sqrt{cr}}{1-2\epsilon}\right)^{2r}\]
colors.
\end{lemma}
\begin{proof}
Because the entries of $M$ are bounded by $c$, \lemmaref{Linial} shows how to find a factorization $XY^T = M$ such that, if $u_i$ and $v_i$ are the row vectors of $X$ and $Y$ respectively, for all $i,j \in [n]$, we have $\|u_i\|,\|v_j\| \leq (cr)^{1/4}$. Note that since $u_i \cdot v_i \geq 1-\epsilon$ for all $i$, this implies a corresponding lower bound of $\|u_i\|,\|v_j\| \geq (1-\epsilon)(cr)^{-1/4}$. 

Now let $W$ be a $\delta$-net of the hypercube of sidelength $2(cr)^{1/4}$ in $\mathbb{R}^r$ centered at the origin (we will pick $\delta$ sufficiently small later). In particular we can pick $W$ so that $|W| \leq (2(cr)^{1/4}/\delta)^r$. \Rnote{Changed the coloring from looking at the union to the cross product.} \eat{Now define the following coloring function $f: V \rightarrow W\times W: f(i) = (w_1,w_2)$ if $u_i,v_i \in B_{\delta}(w_1) \cup B_{\delta}(w_2)$. If $u_i$ or $v_i$ lie in multiple balls of the $\delta$-net, simply pick one arbitrarily (perhaps the closest point $w \in W$). The function $f$ must color every vertex because the $\delta$-net covers every vector. 

We prove that $f$ is a valid coloring if $\delta < (1-\epsilon)(cr)^{-1/4}$. If $i$ and $j$ are the same color, then either $u_i$ and $v_j$ lie in the same ball of radius $\delta$, or $v_j$ lies in the same ball as $v_i$. In the former case, 
\begin{align*}
|\langle u_i,v_j\rangle - \langle u_i,u_i\rangle| &= |\langle u_i,v_j - u_i\rangle| \\
&\leq \|u_i\|\cdot \|v_j-u_i\| \\
&\leq \|u_i\| \cdot \delta \\
&< \langle u_i,u_i\rangle
\end{align*}
where the last inequality follows because $\delta < \|u_i\|$. Thus $u_i \cdot v_j \neq 0$, and so $(i,j) \notin E$.}
Now define the following coloring function $f: V \rightarrow W\times W: f(i) = (w_1,w_2)$ if $u_i \in B_{\delta}(w_1)$ and $v_i \in B_{\delta}(w_2)$. If $u_i$ or $v_i$ lie in multiple balls of the $\delta$-net, simply pick one arbitrarily (perhaps the closest point $w \in W$). The function $f$ must color every vertex because the $\delta$-net covers every vector. 

We prove that $f$ is a valid coloring if $\delta < (1-2\epsilon)(cr)^{-1/4}/2$. If $i$ and $j$ are the same color, $v_j$ lies in the same ball as $v_i$. Then, 
\begin{align*}
|\langle u_i,v_i\rangle - \langle u_i,v_j\rangle| &= |\langle u_i,v_i - v_j\rangle| \\
&\leq \|u_i\|\cdot \|v_i - v_j\| \\
&\leq (cr)^{1/4}(2\delta) \\
&< (1-2\epsilon)
\end{align*}
and since $u_i \cdot v_i \geq 1-\epsilon$, we must have $|u_i \cdot v_j| > \epsilon$, so $(i,j) \notin E$. Now note that the size of the coloring is at most 
\[\chi(G) \leq |W|^2 = \left(\frac{2(cr)^{1/4}}{\delta}\right)^{2r} = \left(\frac{4\sqrt{cr}}{1-2\epsilon}\right)^{2r}\]
and this completes the proof.
\end{proof}
Finally, we can pad the output matrix with zeros just like in the previous section in order to achieve a larger number ($0.9$ fraction) of revealed entries. 
\eat{
The above reduction produces a partial matrix $P_G$ that has $|V| + |E|$ revealed entries, which could be much less than $0.9n^2$ if the graph $G$ is sparse. However, we can simply pad the matrix $P_G$ with zeros, i.e. output the $10|V| \times 10|V|$ matrix
\[P'_G = \left[\begin{tabular}{cc} $P_G$ & $0$ \\ $0$ & $0$\end{tabular}\right].\]}
Combined with \lemmaref{coloringlemma}, the above implies \theoremref{maincomplete}. \Bnote{Put the section on padding with zeros back in}
\section{CSP Reduction for \psdcompletion}

In this section, we include a CSP based hardness reduction for
\psdcompletion that is robust and applicable with noise.

Let $\phi$ be an instance of a \exactonesat with variables $x_1,\dots,x_n$ and clauses $C_1,\dots,C_m$. Recall that a partial PSD matrix is equivalent to a list of inner product constraints. We will describe our reduction in this framework. Let $x_0$ be a "reference variable" unused in $\phi$. For every variable $x \in \{x_0,\dots,x_n\}$, index a set of vectors $\{u_{(x,s)}\}_{s\in I}$ by $I = [2k] \cup \left(\binom{2k}{2}\times \{\pm 1\}\right)$. Then form the \emph{internal variable} constraints
\begin{itemize}
\item $u_{(x,s)} \cdot u_{(x,s)} = 1$ for all $s \in I$.
\item $u_{(x,i)} \cdot u_{(x,j)} = \delta_{ij}$ for $i,j \in [2k]$.
\item $u_{(x,i,j,+1)} \cdot u_{(x,i)} = u_{(x,i,j,+1)} \cdot u_{(x,j)} = \frac{1}{\sqrt{2}}$ for $(i,j) \in \binom{2k}{2}$.
\item $u_{(x,i,j,-1)} \cdot u_{(x,i)} = -u_{(x,i,j,-1)} \cdot u_{(x,j)} = \frac{1}{\sqrt{2}}$ for $(i,j) \in \binom{r_1}{2}$.
\end{itemize}
These constraints force $\{u_{(x,i)}\}_i$ to be a $2k$-dimensional orthonormal basis for any $x$, as well as
\[u_{(x,i,j,+1)} = \frac{1}{\sqrt{2}}\left(u_{(x,i)} + u_{(x,j)}\right),\] and
\[u_{(x,i,j,-1)} = \frac{1}{\sqrt{2}}\left(u_{(x,i)} - u_{(x,j)}\right).\]
Let $p: [2k] \rightarrow [2k]$ be the function $p(i) = i+1$ if $i$ is odd, and $p(i) = i-1$ if $i$ is even. Now for every variable $x \in \{x_1,\dots,x_n\}$, form the \emph{external variable} constraints
\begin{itemize}
\item $u_{(x_0,i)} \cdot u_{(x,j)} = 0$ if $j \neq p(i)$.
\item $u_{(x_0,i,p(i),+1)} \cdot u_{(x,i,p(i),+1)} = 0$ for odd $i \in [2k]$.
\item $u_{(x_0,i,j,+1)} \cdot u_{(x_0,i,j,-1)}$ for $(i,j) \in \binom{2k}{2}$ and $i,j$ odd. 
\end{itemize}
Let $C_0$ be a "reference clause" not referring to any clause of $\phi$. Index a set of vectors $\{u_C\}_C$ by $\{C_0,\dots,C_m\}$. Form the \emph{internal clause} constraints
\begin{itemize}
\item $u_{C_0} \cdot u_{C_0} = 1$.
\item $u_{C_0} \cdot u_{(x_0,2g-1)} = \frac{1}{\sqrt{k}}$ for every $g \in [k]$.

For every clause $C \in \{C_1,\dots,C_m\}$, with variables $\{x_{i_1},x_{i_2},\dots,x_{i_{k}}\}$ with signs $\{s_1,\dots,s_{k}\}$, i.e. each $s_g \in \{+1,-1\}$, 
\item $u_C \cdot u_C = 1$.
\item $u_C \cdot u_{(x_{i_g},2g)} = \frac{s_g}{\sqrt{k}}$
for every $g \in [k].$
\end{itemize}
Finally, for each clause $C \in \{C_1,\dots,C_m\}$, the \emph{external clause} constraints are 
\begin{itemize}
\item $u_{C_0} \cdot u_C = (1-2/k)$ for every $C \in \{C_1,\dots,C_m\}$.
\end{itemize}
Then we have the following theorem:
\begin{theorem}\theoremlabel{mainwitherrors}
If there is a satisfying assignment to $\phi$, then there is a set of vectors lying in $\mathbb{R}^{2k}$ satisfying the above constraints exactly. Conversely, if $\{u_{(x,s)}\}_{s\in I}$ and $\{u_{C_j}\}_{j \in [m]}$ are vectors in $\mathbb{R}^r$ that satisfy the constraints up to an additive $\pm \epsilon$, and $r \leq 4k-1$ and $\epsilon < 10^{-6}k^{-5}$, then there is a satisfying assignment to $\phi$.
\end{theorem}
\begin{proof}
To start, we prove completeness. Let $f$ be a satisfying assignment to $\phi$. Then we propose the following vectors: \begin{itemize}
\item $u_{(x_0,i)} = e_i$, where $e_i$ is the $i$th standard basis vector. 
\item For every $x \in \{x_1,\dots,x_n\}$, set $u_{(x,p(i))} = f(x)e_i$ for $i$ odd and $u_{(x,p(i))} = -f(x)e_i$ for $i$ even.  For every $i,j \in \binom{2k}{2}$, set $u_{(x,i,j,+1)}$ and $u_{(x,i,j,-1)}$ to the normalized sum and difference of $u_{(x,i)}$ and $u_{(x,j)}$. 
\item $u_{C_0} = \frac{1}{\sqrt{k}}\sum_g u_{(x_0,2g-1)}$.
\item For every $C \in \{C_1,\dots,C_m\}$ with variables $\{x_{i_1},\dots,x_{i_k}\}$ and signs $\{s_1,\dots,s_k\}$, set $u_C = \frac{1}{\sqrt{k}}\sum_g s_gu_{(x_{i_g},2g)}$.
\end{itemize}. 
It is clear that all internal constraints are satisfied. The external variable constraints are also easy to verify. For any clause $C$,
\begin{align*}
u_C \cdot u_{C_0} &= \frac{1}{k}\sum_{g,g'=1}^k s_g(u_{(x_{i_g},2g)} \cdot u_{(x_0,2g'-1)}) \\
&= \frac{1}{k}\sum_{g,g'=1}^k s_gf(x_{i_g})(e_{2g-1} \cdot e_{2g'-1}) \\
&= \frac{1}{k}\sum_{g=1}^k s_gf(x_{i_g}).
\end{align*}
Since $f$ is a satisfying assignment to $\phi$, exactly one variable in the clause is $-1$, thus the sum is exactly $(k-2)$, and so $u_C \cdot u_{C_0} = (1-2/k)$. 

To prove soundness, we will start by assuming that all internal constraints are satisfied exactly, and only the external constraints contain errors. We will decode a satisfying assignment to $\phi$ under this assumption. Then we will show how to take the initial set of vectors and adjust them slightly to get a set of vectors perfectly satisfying the internal constraints. 
\begin{lemma}
Fix $r < 4k$. For each $x \in \{x_0,x_1,\dots,x_n\}$, let $\{u_{(x,s)}\}_{s\in I}$ be a set of vectors in $\R^r$ satisfying the internal variable constraints exactly, and assume every external variable constraint is satisfied up to a small additive $\pm \delta$ such that $\delta < 1/12k$. Then for any $x \in \{x_1,\dots,x_n\}$ and odd $i,i' \in [2k]$, 
\[\text{sign}(u_{(x_0,i)} \cdot u_{(x,p(i))}) = \text{sign}(u_{(x_0,i')} \cdot u_{(x,p(i'))})\]
and
\[1 \geq |u_{(x_0,i)} \cdot u_{(x,p(i))}| \geq 1-12\delta k\]
\end{lemma}
\begin{proof}
Because the internal constraints are satisfied, $T_0 = \{u_{(x_0,i)}\}_{i \in [2k]}$ and $T_x = \{u_{(x,i)}\}_{i \in [2k]}$ are orthonormal bases, so there is an orthonormal transformation $Q: \R^r \rightarrow \R^r$ such that $Q(u_{(x_0,i)}) = u_{(x,i)}$. We write $Q$ in any basis containing $T_0$:
\[Q = \left[\begin{tabular}{cc} $Q'$ & $A$ \\ $B$ & $C$\end{tabular}\right]\]
where $Q'$ is a transformation from $T_0$ to itself. Because $|u_{(x_0,i)} \cdot u_{(x,j)}| \leq \delta$ for any $j \neq p(i)$, $Q'$ is at most $\delta$ except on the $2 \times 2$ block diagonal. Now $|u_{(x_0,i,p(i),+1)} \cdot u_{(x,i,p(i),+1)}| \leq \delta$ implies that $|Q(i,p(i)) - Q(p(i),i)| \leq 3\delta$. Finally, $|u_{(x_0,i,i',+1)} \cdot u_{(x,p(i),p(i')-1)}| \leq \delta$ for odd $i$ and $i'$ implies that $|Q(i,p(i)) - Q(i',p(i'))| \leq 3\delta$. Because of these conditions, we can write $Q' = R + S$, where 
\[R = R_1 \oplus R_2 \oplus \dots \oplus R_{k}\]
and
\[R_g = \left[\begin{tabular}{cc} $0$ & $a$ \\ $-a$ & $0$\end{tabular}\right]\text{ for every $g \in [k]$}\]
and $|S(i,j)| \leq 6\delta$ for every $i,j$. Then for any unit vector $x \in \text{span}(T_0)$, 
\[\|Q'x\| \geq \|Rx\| - \|Sx\| \geq |a| - 6\delta k.\] 
In particular, let $x$ be a null vector of $B$, which exists because $B$ is an $(r - 2k) \times 2k$ matrix, and $r < 4k$. Then 
\[Q\left[\begin{tabular}{c} $x$ \\ $0$\end{tabular}\right] = \left[\begin{tabular}{c} $Q'x$ \\ $0$\end{tabular}\right]\]
and since $Q$ is an orthogonal transformation, this implies $\|Q'x\| = 1$, and thus $|a| \geq 1-6\delta k$. Now recalling the definition of $Q$, for any odd $i$, $u_{(x_0,i)} \cdot u_{(x,p(i))} = Q(i,p(i)) = a + S(i,p(i))$, and thus for any odd $i$
\[1 \geq |u_{(x_0,i)} \cdot u_{(x,p(i))}| \geq 1 - 12\delta k\]
and the sign is the same for any odd $i$ if $\delta < \frac{1}{12\delta k}$. The upper bound follows from the fact that $u_{(x_0,i)}$ and $u_{(x,p(i))}$ are unit vectors. 
\end{proof}  
We take the interpretation that the variable $x$ is set to sign$(u_{(x_0,1)}\cdot u_{(x,2)})$. Now we prove that this assignment must be a satisfying assignment because of the clause constraints.
\begin{lemma}
For each $x \in \{x_0,\dots,x_n\}$, let $\{u_{(x,s)}\}_{s\in I}$ satisfy the assumptions of the previous lemma, and for each $C \in \{C_1,\dots,C_m\}$ with variables $\{x_{i_1},x_{i_2},\dots,x_{i_{k}}\}$ with signs $\{s_1,\dots,s_{k}\}$, let 
\[u_C = \frac{1}{\sqrt{k}}\sum_{g=1}^k s_gu_{(x_{i_g},2g)}\]
and let
\[u_{C_0} = \frac{1}{\sqrt{k}}\sum_{g=1}^k u_{(x_0,2g-1)}.\] 
Then if the constraints $u_{C_0} \cdot u_C = (1-2/k)$ are satisfied up to an additive $\pm \delta$ and
\[\delta < \min\left(\frac{2}{13k^2},\frac{2}{24k+k^2}\right)\]
then the assignment $f(x) = \text{sign}(u_{(x_0,1)} \cdot u_{(x,2)})$ is a satisfying assignment.
\end{lemma}
\begin{proof}
Let $C \in \{C_1,\dots,C_m\}$. From the definitions of $u_C$ and $u_{C_0}$,   
\begin{align*}
u_{C_0} \cdot u_C &= \frac{1}{k} \sum_{g,g'=1}^{k} s_g\left(u_{(x_0,2g'-1)} \cdot u_{(x_{i_g},2g)}\right) \\
&= \frac{1}{k}\sum_{g=1}^k s_g\left(u_{(x_0,2g-1)}\cdot u_{(x_{i_g},2g)}\right) + \frac{1}{k}\sum_{g\neq g'}^k s_g(u_{(x_0,2g'-1)} \cdot u_{(x_{i_g},2g)})
\end{align*}
where $x_{i_g}$ are the variables appearing in clause $C$. We argue that $|u_{C_0}\cdot u_C - (1-4/r_1)| < \delta$ implies that $f$ is satisfying. Note that the dot products in first have magnitudes between $1$ and $1-24\delta k$, and those in the second sum have magnitudes at most $\delta$. If $\delta < 2/(24k+k^2)$, even if $k-2$ of the dot products have sign $+1$ and $2$ have sign $-1$,
\begin{align*}
(1-4/r_1) - u_{C_0} \cdot u_C &\geq (1-4/r_1) - \frac{1}{k}\left[(k - 2) - 2(1-12\delta k) + k(k-1)\delta)\right] \\
&= 2/k - \delta(23+k) \\
&> \frac{2}{24k+k^2} \\
&> \delta
\end{align*}
Now if $\delta < 2/13k^2$, then even if all $k$ of the dot products have sign $+1$
\begin{align*}
u_{C_0} \cdot u_C - (1-4/r_1) &\geq \frac{1}{k}\left[k(1-12\delta k) - k(k-1)\delta\right] - (1-2/k) \\
&= 2/k - \delta(13k-1) \\
&> \frac{2}{13k^2} \\
&> \delta.
\end{align*} 
These two facts mean that $|u_{C_0} \cdot u_C - (1-2/k)| < \delta$ implies that exactly one of the dot products in the first sign can have sign $-1$ and the rest must have sign $+1$, i.e. that $f$ satisfies the clause $C$. This is true for every clause, so $f$ must be a satisfying assignment.
\end{proof}
These lemmas prove \theoremref{psdapproxcompletemain} if only the external variable constraints experience errors, and $u_C$ and $u_{C_0}$ are constructed properly. The next sequence of lemmas show how to transform the problem from errors on all constraints into errors on only external constraints.
\begin{lemma}
Let $\{u_{(x,s)}\}_{s \in I}$ be a set of vectors satisfying the internal constraints to within an additive $\pm \epsilon$ for $\epsilon < O(1/k)$. Then there is a set of vectors $\{\tilde{u}_{(x,s)}\}_{s\in I}$ satisfying the internal constraints exactly such that 
\[ \|\tilde{u}_{(x,s)} - u_{(x,s)} \| \leq 3\sqrt{\epsilon}.\]
\end{lemma}
\begin{proof}
Let $u_{(x,1)} = u_1^\perp + u_1^\parallel$, where $u_1^\perp$ is the part of $u_{(x,1)}$ perpendicular to the subspace $\text{span}(\{u_{(x,i)}\}_{i \neq 1})$, and $u_1^\parallel$ is the part parallel. Note that since $|u_{(x,1)} \cdot u_{(x,i)}| \leq \epsilon$ for any $i \neq 1$, 
\[\|u_1^\parallel\| \leq \sqrt{2k-1}\frac{\epsilon}{\sqrt{1-\epsilon}}.\]
Let $\tilde{u}_{(x,1)} = u_1^\perp/\|u_1^\perp\|$. Then
\begin{align*}
\|\tilde{u}_{(x,1)} - u_{(x,1)}\| &= \|u_1^\perp(1/\|u_1^\perp\| - 1) - u_1^\parallel\| \\
&\leq \left(1/\|u_1^\perp\| - 1\right)\|u_1^\perp\| + \|u_1^\parallel\| \\
&= 1 - \sqrt{\|u_{(x,1)}\|^2-\|u_1^\parallel\|^2} + \|u_1^\parallel\| \\
&\leq 1 - \sqrt{1-\epsilon} + 2\sqrt{2k-1}\frac{\epsilon}{\sqrt{1-\epsilon}} \\
&\leq \epsilon(1+2\sqrt{k})
\end{align*}
where the last inequality uses a series approximation. Now proceeding inductively, let $u_{(x,i)} = u_i^\perp + u_i^\parallel$, where the subspace considered is $\text{span}(\{\tilde{u}_{(x,j)}\}_{j < i}) \cup \text{span}(\{u_{(x,j)}\}_{j > i})$, and we obtain an identical bound on the difference of norms. Note that the $\{\tilde{u}_{(x,i)}\}_{i \in r_1}$ satisfy the internal constraints, and define the remaining vectors to be the sums forced by the remaining constraints. 

To bound the dot products of sums of basis vectors, we first bound
\begin{align*}
\Big\|\frac{1}{\sqrt{2}}(u_{(x,i)} &+ u_{(x,j)}) - u_{(x,i,j,+1)}\Big\| = \\
&=\left(\frac{1}{\sqrt{2}}(u_{(x,i)} + u_{(x,j)}) - u_{(x,i,j,+1)}\right)\cdot \left(\frac{1}{\sqrt{2}}(u_{(x,i)} + u_{(x,j)}) - u_{(x,i,j,+1)}\right)^{1/2} \\
&\leq\left(2(1+\epsilon) + \epsilon - 2\sqrt{2}(1/\sqrt{2}-\epsilon)\right)^{1/2} \\
&\leq \sqrt{\epsilon}\sqrt{3 + 2\sqrt{2}}
\end{align*}
and now
\begin{align*}
\left\Vert\tilde{u}_{(x,i,j,+1)} - u_{(x,i,j,+1)}\right\Vert &= \left\Vert\frac{1}{\sqrt{2}}(\tilde{u}_{(x,i)} + \tilde{u}_{(x,j)}) - u_{(x,i,j,+1)}\right\Vert \\
&\leq \left\Vert\frac{1}{\sqrt{2}}(\tilde{u}_{(x,i)} + \tilde{u}_{(x,j)}) - \frac{1}{\sqrt{2}}(u_{(x,i)} - u_{(x,j)})\right\Vert + \left\Vert u_{(x,i,j,+1)} - \frac{1}{\sqrt{2}}(u_{(x,i)} + u_{(x,j)})\right\Vert \\
&\leq \epsilon\sqrt{2}(1+2\sqrt{k}) + \sqrt{\epsilon}\sqrt{3+2\sqrt{2}} \\
&\leq 3\sqrt{\epsilon}
\end{align*}
where the last inequality follows because, since $\epsilon < O(1/k)$, $\epsilon\sqrt{k}$ is dominated by $\sqrt{\epsilon}$, and we simply round the coefficients to integers. Note that this also means $3\sqrt{\epsilon} > \epsilon(1+2\sqrt{k})$.  
\end{proof}

\begin{lemma}\lemmalabel{finalsumbound}
For any two $x,y \in \{x_0,\dots,x_n\}$ and $s \in I$, 
\[|\tilde{u}_{(x,s)} \cdot \tilde{u}_{(y,s)} - u_{(x,s)} \cdot u_{(y,s)}| \leq 7\sqrt{\epsilon}\]
\end{lemma}
\begin{proof}
To compress notation, let $u_{(x,s)} = u$ and $u_{(y,s)} = v$ and likewise for the tilde versions.
\begin{align*}
|\langle \tilde{u} - u,\tilde{v} - v\rangle| &= |\langle \tilde{u},\tilde{v}\rangle + \langle u,v\rangle - \langle \tilde{u},v\rangle - \langle u,\tilde{v}\rangle| \\
&\geq |\langle \tilde{u},\tilde{v}\rangle - \langle u,v\rangle| - |\langle \tilde{u},v\rangle - \langle u,v\rangle| - |\langle u,\tilde{v}\rangle - \langle u,v\rangle| \\
&\geq |\langle \tilde{u},\tilde{v}\rangle - \langle u,v\rangle| - \|\tilde{u}-u\| - \|\tilde{v}-v\| 
\end{align*}
And note that 
\[|\langle \tilde{u}-u,\tilde{v}-v\rangle| \leq \|u-\tilde{u}\|\|v-\tilde{v}\|\]
and thus
\begin{align*}
|\langle \tilde{u},\tilde{v}\rangle - \langle u,v\rangle| &\leq \|u-\tilde{u}\|\|v-\tilde{v}\| + \|u-\tilde{u}\| + \|v-\tilde{v}\| \\
&\leq 9\epsilon + 6\sqrt{\epsilon} \\
&\leq 7\sqrt{\epsilon}
\end{align*}
\end{proof}
\begin{lemma}
For each clause $C \in \{C_1,\dots,C_m\}$ containing variables $\{x_{i_1},\dots,x_{i_{r_1/2}}\}$, let $\tilde{u}_C = \frac{1}{\sqrt{k}}\sum_g s_g\tilde{u}_{(x_{i_g},2g)}$, and let $\tilde{u}_{C_0} = \frac{1}{\sqrt{k}}\sum_g \tilde{u}_{(x_{i_g},2g-1)}$. Then for any $C \in \{C_1,\dots,C_m\}$, 
\[|\langle \tilde{u}_{C_0},\tilde{u}_C\rangle - \langle u_{C_0},u_C\rangle| \leq 35\sqrt{\epsilon 2k}\]
\end{lemma}
\begin{proof}
For any clause $C$ containing variables $\{x_{i_1},\dots,x_{i_{k}}\}$, we bound the norm,
\begin{align*}
\left\Vert u_C - \frac{1}{\sqrt{k}}\sum_{g=1}^{k} s_gu_{(x_{i_g},2g)}\right\Vert &\leq \left[\left(u_C - \frac{1}{\sqrt{k}}\sum_{g=1}^{k} s_g u_{(x_{i_g},2g)}\right) \cdot \left(u_C - \frac{1}{\sqrt{k}}\sum_{g=1}^{k} s_g u_{(x_{i_g},2g)}\right)\right]^{1/2} \\
&\leq \left(1+\epsilon +k \cdot\frac{1}{k}(1+\epsilon) + \frac{1}{k} \cdot k(k-1)\epsilon - 2k\cdot\frac{1}{\sqrt{k}}(1/\sqrt{k}-\epsilon)\right)^{1/2} \\
&\leq (\epsilon(1+k+2\sqrt{k}))^{1/2} \\
&\leq \sqrt{2k\epsilon}
\end{align*}
and
\begin{align*}
\left\Vert\tilde{u}_C - \frac{1}{\sqrt{k}}\sum_{g=1}^{k} s_gu_{(x_{i_g},2g)}\right\Vert &\leq \frac{1}{\sqrt{k}}\sum_{g=1}^{k} \norm{\tilde{u}_{(x_{i_g},2g)} - u_{(x_{i_g},2g)}} \\
&\leq \epsilon(2k+\sqrt{k})
\end{align*}
and thus by triangle inequality,
\[\|\tilde{u}_C - u_C\| \leq \sqrt{2k\epsilon} + \epsilon(2k+\sqrt{k}) \leq 2\sqrt{k\epsilon}\]
because $k\epsilon < \sqrt{k\epsilon}$. The identical calculation bounds $\|\tilde{u}_{C_0} - u_{C_0}\|$. Finally, just as in the proof of Lemma \lemmaref{finalsumbound}, for any clause $C$, we have 
\begin{align*}
|\langle \tilde{u}_{C_0},\tilde{u}_{C}\rangle - \langle u_{C_0},u_{C}\rangle &\leq \|\tilde{u}_{C_0}-u_{C_0}\|\|\tilde{u}_{C} - u_{C}\| + \|\tilde{u}_{C_0}-u_{C_0}\| + \|\tilde{u}_{C} - u_{C}\| \\
&\leq 4k\epsilon + 4\sqrt{k\epsilon} \\
&\leq 5\sqrt{k\epsilon}
\end{align*}
\end{proof}
If every constraint experiences error at most $\epsilon$, then we can construct an alternate solution that satisfies the internal constraints exactly and every external constraint experiences error at most $\delta \leq 5\sqrt{k\epsilon}$. Since we require at most
\[\delta \leq \min\left(\frac{1}{24k},\frac{2}{25k^2},\frac{2}{48k+k^2}\right)\]
error on the external constraints, we can handle error at most $\epsilon < O(k^{-5})$. This completes the proof of the theorem.
\end{proof}

\end{document}